\documentclass[conference]{IEEEtran}
  \usepackage[pdftex]{graphicx}
\usepackage{amsmath,amssymb,amsthm, latexsym}
\usepackage{caption}
\usepackage{hyperref}

\usepackage[english]{babel}
\usepackage[utf8]{inputenc}
\usepackage{algorithm}
\usepackage[noend]{algpseudocode}
\usepackage{xcolor}

\usepackage[utf8]{inputenc}
\usepackage[english]{babel}
 
\newtheorem{theorem}{Theorem}
 \newtheorem{lemma}{Lemma}
 \newtheorem{definition}{Definition}

\def\thesection{\Roman{section}}
\def\thesubsection{\Roman{section}.\Roman{subsection}}

\begin{document}
%
\title{A H-ARQ scheme for polar codes}
\author{\IEEEauthorblockN{Saurabha R. Tavildar}
\IEEEauthorblockA{Email: tavildar at gmail }
}

\maketitle
\begin{abstract}
~We consider the problem of supporting H-ARQ with polar codes. For supporting H-ARQ, we propose to create redundancy versions based on different, but equivalent, subsets of a polar code. The equivalent subsets are created from an initial subset of a polar code using an inherent symmetry in polar code construction. A greedy construction is used to create the initial subset of a polar code. 

	We demonstrate performance of proposed constructions via simulations for binary input AWGN channel. We demonstrate that a (4096, 1024) polar code can be divided into two disjoint (2048, 1024) subset polar codes, which when decoded individually are within 0.2 dB (at $1 \%$ BLER) of a (2048, 1024) polar code, and achieve performance of a (4096, 1024) polar code when decoded jointly. 
\end{abstract}

\section{Introduction}
Polar codes, introduced in \cite{Arikan}, were proved to achieve the symmetric capacity for BDMCs. The original construction in \cite{Arikan} is defined for block length values that are a power of $2$. In this paper, we consider ``subset polar codes" which are constructed by puncturing a subset of coded bits from a polar code. For construction of subset polar codes, we start with a low rate polar code, and greedily puncture output bits (similar to \cite{Khamy}) to create a code of higher rate {\it without  re-optimizing} the set of information bits. For supporting H-ARQ, it is important that the set of information bits for subset polar code is same as the original polar code. This is because re-optimizing the set of information bits changes the code structure which makes it difficult for the receiver to jointly decode multiple transmissions. 

We use the notation $(N \succeq M, K)$ to denote a subset polar code of block length $M$ constructed from an $(N, K)$ polar code by puncturing $N-M$ coded bits. We call the $(N, K)$ polar code as a mother polar code. In this paper, we propose a construction for subset polar codes, and its use for H-ARQ. The main simulation results regarding subset polar code constructions and H-ARQ are in \hyperref[fig:harq_performance]{Figure~\ref*{fig:harq_performance}} which show:
\begin{itemize}
\item Two $(4096 \succeq 2048, 1024)$ subset polar codes, decoded \textit{individually}, perform within $0.2$ dB (at $1 \%$ BLER) of a $(2048, 1024)$ polar code.
\item Two $(4096 \succeq 2048, 1024)$ subset polar codes, decoded \textit{jointly}, perform $0.75$ dB better (at $1 \%$ BLER) than a $(2048, 1024)$ polar code and achieve, by construction, performance of the $(4096, 1024)$ mother polar code that they were constructed from. 
\end{itemize}
Throughout this paper, the following assumptions are used: (i) code construction in \cite{Arikan} based on evaluation of Bhattacharya bounds is used. Under this construction, optimized value (for $1 \%$ BLER) of $
\epsilon$ is used for (mother) polar code construction which corresponds to $\epsilon = 0.64$ for the $(4096, 1024)$ polar code and $\epsilon = 0.32$ for the $(2048, 1024)$ polar code; (ii) for decoding, simplified LLR based, CRC-aided (16 bit), list decoding ($L = 32$) algorithm is used \cite{Tal, Stimming}. \footnote{Performance of the $(2048, 1024)$ polar code is 0.1 dB worse than reported in \cite{Tal} possibly due to sub-optimal code construction and/or LLR based receiver which uses the ‘hardware-friendly’ (see \cite{Stimming}) update equations for large values of LLR. C and MATLAB implementations are provided in \cite{Tavildar}.}
\begin{figure}[!tb] 
\centering
\includegraphics[scale=.42]{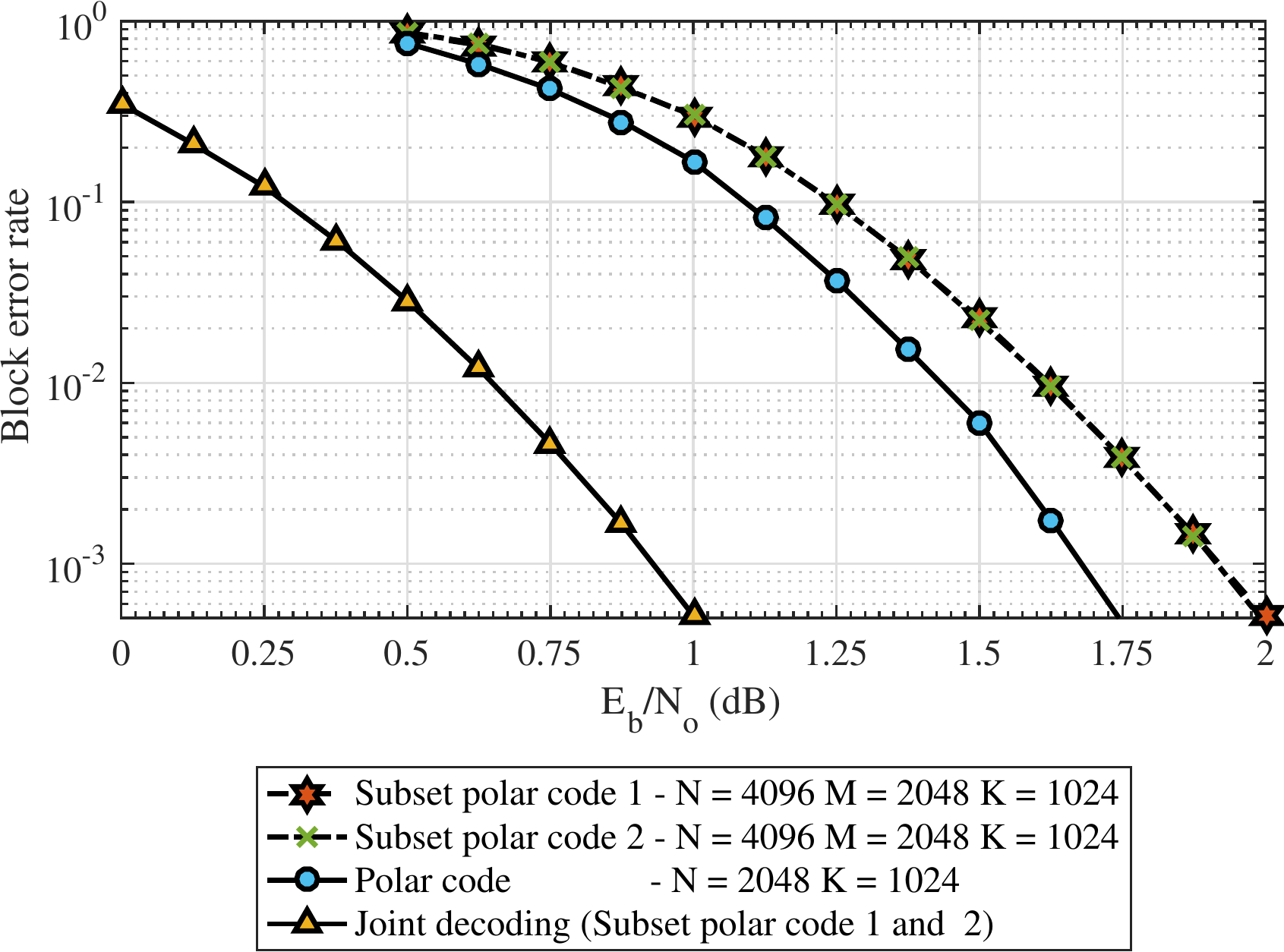}
\caption{Subset polar codes performance.}\label{fig:harq_performance}
\end{figure}

The rest of the paper is organized as follows: notation is discussed in \hyperref[sec:notation]{Section~\ref*{sec:notation}}. \hyperref[sec:code_construction]{Section~\ref*{sec:code_construction}} discusses subset polar code construction. In \hyperref[sec:harq]{Section~\ref*{sec:harq}}, we provide a H-ARQ extension for subset polar codes. In \hyperref[sec:prior_work]{Section~\ref*{sec:prior_work}}, we discuss relation with prior work, and finally in \hyperref[sec:equivalent]{Section~\ref*{sec:equivalent}}, we prove some results regarding subset polar codes that justify constructions proposed in this paper.

\section{Notation}\label{sec:notation}
\begin{itemize}
\item $N$: block length of mother polar code ($N = 2^n$); 
\item $M$: block length of a subset polar code ($M \leq N$); 
\item $K$: number of information bits; 
\item $P = \{p_1,~p_2, ... ,~p_{N-M}\}, 1 \leq p_i \leq N~\forall~i$ denotes a puncturing pattern of distinct indexes that specifies a subset polar code by removing coded bits with indexes $p_i$ from a $(N, K)$ mother polar code;
\item $S = \{s_1,~s_2, ... ,~s_{M}\}, 1 \leq s_i \leq N~\forall~i$ denotes a subset polar code by using coded bits with distinct indexes $s_i$ from a $(N, K)$ mother polar code; 
\item $x_nx_{n-1}...x_1$ denotes binary representation of $x-1$ where $1 \leq x \leq N$;
\item Let $1 \leq x,~y,~z \leq N$ be integers, define $z = x~\underline{\oplus}~y$ if $z_i = x_i\oplus y_i~\forall~i,~1\leq i \leq n$ where $x_i,~y_i,~z_i$ are binary representations of $x-1,~y-1,~z-1$ respectively, and $\oplus$ denotes the binary XOR operation;
\item $\epsilon$ is the BEC parameter used for code construction.
\end{itemize}

\section{Subset polar code construction}\label{sec:code_construction}

For construction of subset polar codes, we start with a low rate mother polar code. Polar code construction in \cite{Arikan} based on evaluation of Bhattacharya parameter bounds is used. The set of information and CRC bits is optimized for $\epsilon$ parameter that gives the best performance for the mother polar code. A subset polar code is constructed by puncturing coded bits from the mother polar code. The puncturing algorithm is a greedy algorithm that selects a coded bit to puncture at each step by estimating the BLER after puncturing that coded bit in addition to already punctured coded bits. The algorithm additionally adaptively changes the design $\epsilon$ to keep the estimated BLER at a given target BLER. The proposed algorithm is similar to the PPA algorithm in \cite{Khamy} with the exception of update of the design $\epsilon$. The update of the design $\epsilon$ appears to improve performance by keeping the union bound tight for estimating BLER (see results in \hyperref[sec:prior_work]{Section~\ref*{sec:prior_work}}). The above description is provided in the form of a pseudo-code in \hyperref[alg:puncture]{Algorithm \ref*{alg:puncture}} and \hyperref[alg:eps]{Algorithm \ref*{alg:eps}}. The algorithm for estimating BLER is omitted. BLER is estimated the union bound, which uses the sum of Bhattacharya bounds for the polarized channels corresponding to information and CRC bits of the mother polar code.

\begin{figure}[!hb]
\centering
\includegraphics[scale=.42]{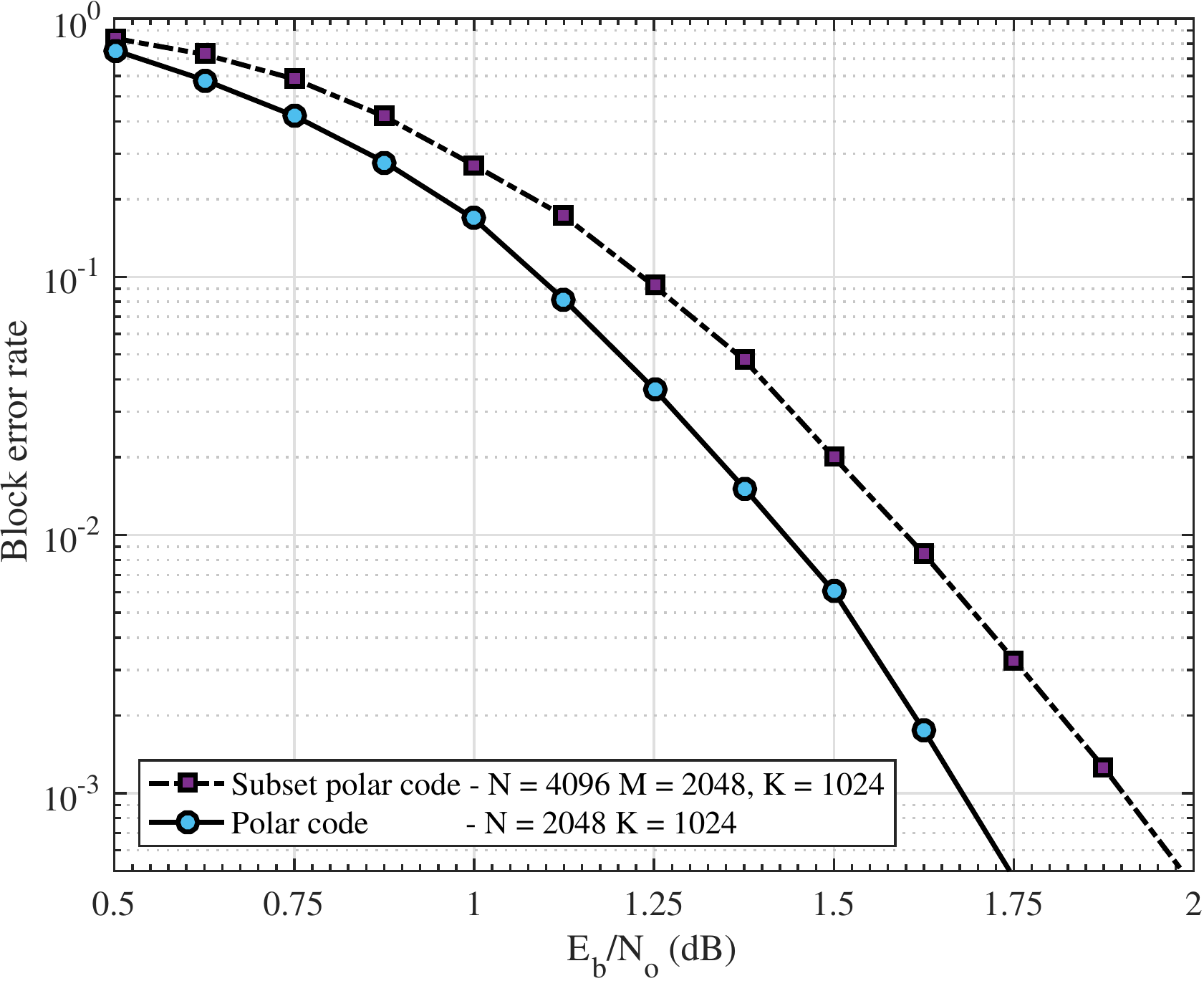}
\caption{Subset polar code vs polar code}\label{fig:polar_performance}
\end{figure}

The performance of the scheme was evaluated for codes with $K = 1024$ by starting with a $(4096, 1024)$ code and constructing subset polar codes with smaller block length. \hyperref[fig:polar_performance]{Figure~\ref*{fig:polar_performance}} shows the comparison of  subset polar code and a code designed directly for $N = 2048$. We see that the subset code is less than $0.2$ dB away at $10^{-2}$ block error rate. We that some performance loss is expected since the set of information bits for the subset polar code is optimized for the mother polar code ($N = 4096$). As we will see in \hyperref[sec:harq]{Section~\ref*{sec:harq}} that this loss can be recovered when considered in context of H-ARQ. 

\begin{algorithm}
\caption{Adaptive greedy construction}\label{alg:puncture}
\begin{algorithmic}[1]
\Procedure{Puncture}{$M$}\label{proc:puncture}\Comment{M = length of subset code}
\State $P \gets \varnothing$
\State $\epsilon \gets \text{Design epsilon}$
\State $e \gets \textsc{Evaluate Bler}(P, \epsilon) $
\For{$ i = 1 : N - M$}
\State $E \gets N * \text{ones}(N)$ 
\For{$ l = 1 : N$}
\If {$ l ~ \text{in} ~ P$}
\State \texttt{continue}
\EndIf
\State $E[l] = \textsc{Estimate Bler}(P \cup \{l\}, \epsilon)$
\EndFor
\State $P = P \cup \text{argmin}(E)$
\State $\epsilon = \textsc{Update Eps}(P, \epsilon, e)$
\EndFor
\State \textbf{return} $P$
\EndProcedure
\end{algorithmic}
\end{algorithm}

\begin{algorithm}
\caption{Epsilon update}\label{alg:eps}
\begin{algorithmic}[1]
\Procedure{UpdateEps}{$P, \epsilon, e$} 
\State $\epsilon_i = 0.001$ \Comment{A small number}
\State $e_o \gets \textsc{Estimate Bler}(P, \epsilon)$
\While{$e_o > e$} 
	\State $e_o \gets \textsc{Estimate Bler}(P, \epsilon = \epsilon - \epsilon_i)$
	\EndWhile
\State \textbf{return} $\epsilon$ 
\EndProcedure
\end{algorithmic}
\end{algorithm}

\section{Polar H-ARQ design and performance}\label{sec:harq}

\subsection{H-ARQ design principles}

We study the problem of supporting H-ARQ with polar codes (called polar H-ARQ) motivated by wireless system design. As background, note that in LTE (see e.g. \cite{212}), H-ARQ is supported for Turbo codes by arranging the coded bits in a circular buffer, and redundancy versions (RV) are specified by an offset within the circular buffer. Here, we look at the problem of constructing RVs for polar codes.  

We consider the following two to be desirable principles of RV design for a wireless system:
\begin{enumerate}
\item{\textit{Individually} decoding each RV has good performance}
\item{\textit{~~Jointly~~} decoding multiple RVs has good performance}
\end{enumerate}

The first principle is motivated by the fact that in a wireless system, different transmissions may experience different channel fades, and hence the performance at the receiver may be dominated by a (re)-transmission that experiences a good channel fade. This is especially important in case of limited feedback system (e.g. only ACK/NACK feedback from the receiver rather than complete channel state information). This motivates the principle that each RV has good decoding performance when individually decoded. The second principle is targeting coding gain from multiple transmissions.

One approach for polar H-ARQ is to take ordering given by a subset polar code, and write it in a circular buffer similar to LTE. However, it is unclear what offsets, if any, provide good performance for other RVs. For example, it is unclear if complement of a subset polar code is a good subset polar code. Simulation results suggest that it is not a good subset polar code. Hence, we propose an alternate way to generate redundancy versions which is discussed next. 

\subsection{Equivalent subset polar codes}

We define a notion of equivalent subset polar codes as:
\begin{definition}\label{def1}
Let $S = \{ s_1, s_2, ... , s_M\}$ and $T = \{ t_1, t_2, ... , t_M\}$  be two $(N \succeq M, K)$ subset polar codes. We say $S$ is equivalent to $T$ if $s_i = t_i~\underline{\oplus}~x~\forall~i$, $1 \leq i \leq M$ for some integer $x$, where $1 \leq x \leq N$.
\end{definition}

The operation $\underline{\oplus}$ is defined for integers in \hyperref[sec:notation]{Section~\ref*{sec:notation}}. With some abuse of notation, we use $S = T~\underline{\oplus}~x$ to denote this relation between subsets, puncturing patterns or subset codes $S$ and $T$. Next, we define equivalence of channels: 
\begin{definition}
Let $W_1$ and $W_2$ be two binary input channels with output alphabets $Y_1$ and $Y_2$ respectively. We define the two channels are equivalent, $W_1 \sim W_2$, if there exists an invertible function $f: Y_1 \rightarrow Y_2$ such that $W_1(y_1|x) = W_2(f(y_1)|x)$.
\end{definition}

We now relate the two notions of equivalence through the following theorem. Let $S$ be a $(N \succeq M, K)$ subset polar code and $W$ be a symmetric B-DMC. Let $W_{N,S}^{(i)}$ denote polarized bit-channels corresponding to the polarization transform for $M$ i.i.d.\ realizations of channel $W$ for indexes in $S$, and $N-M$ realizations of an erasure channel for indexes not in $S$ (defined formally in \hyperref[sec:equivalent]{Section~\ref*{sec:equivalent}}).  

\begin{theorem}\label{thm1}
If $S, T$ are equivalent subset polar codes, then channels $W_{N,S}^{(i)}$ and $W_{N,T}^{(i)}$ are equivalent for each $1 \leq i \leq N$.
\end{theorem}
Proof is given in \hyperref[sec:equivalent]{Section~\ref*{sec:equivalent}}. We note that the notion of equivalent subset polar codes is similar to the equivalent shorterning patterns discussed in \cite{Vera} to reduce complexity of subset code construction. Equivalent subsets will be considered equivalent patterns as per definition in \cite{Vera}. We give an explicit construction (as per Definition 1), and a general proof for this construction. Proof in \cite{Vera} is for equality of the error probability estimate via Gaussian approximation for the AWGN channel.

Finally, we propose to support polar H-ARQ by starting with an initial subset polar code, and creating multiple RVs by selecting an appropriate value $x$ for each RV, and using the construction given by \hyperref[def1]{Definition~\ref*{def1}}. However, we discuss a modification of \hyperref[alg:puncture]{Algorithm~\ref*{alg:puncture}} in order to improve performance of polar H-ARQ under this proposal. 

\subsection{Modification to initial construction of subset polar code}

The proposal to create multiple RVs by selecting $x$ satisfies principle 1 but not necessarily principle 2. One reason for this is  that two equivalent subsets $S$ and $T$ may have significant overlap, and hence $S \cup T$ may not be a good code. To solve this problem, we modify the construction of the initial subset polar code to take into account $x$ while designing the initial pattern $P$. In particular, pattern $P$ is defined while making sure that $P$ and $P~\underline{\oplus}~x$ are disjoint as long as $M \geq N/2$. This is a small modification to \hyperref[alg:puncture]{Algorithm~\ref*{alg:puncture}}, and is shown in \hyperref[alg:sym_puncture]{Algorithm~\ref*{alg:sym_puncture}} below as ``Symmetric greedy construction".

\begin{algorithm}
\caption{Symmetric greedy construction}\label{alg:sym_puncture}
\begin{algorithmic}[1]
\Procedure{Puncture}{$M$}\Comment{M = length of subset code}
\State $P \gets \varnothing$
\State $e \gets \textsc{Estimate Bler}(P, \epsilon) $
\For{$ i = 1 : N - M$}  
\State $E \gets N * \text{ones}(N)$

\For{$ l = 1 : N$}
\If {$ l ~ \text{in} ~ P ~ \textcolor{red}{\text{or} ~ l~\text{in} ~ P ~\underline{\oplus}~x}$}
\State \texttt{continue}
\EndIf
\State $E[l] = \textsc{Estimate Bler}(P \cup \{l\}, \epsilon)$
\EndFor
\State $P = P \cup \text{argmin}(E)$
\State $\epsilon = \textsc{Update Eps}(P, \epsilon, e)$
\EndFor
\State \textbf{return} $P$
\EndProcedure
\end{algorithmic}
\end{algorithm}

\hyperref[fig:symmetric]{Figure~\ref*{fig:symmetric}} shows the performance comparison. The symmetric greedy construction curve is denoted by S-subset polar code (\hyperref[alg:sym_puncture]{Algorithm~\ref*{alg:sym_puncture}}) and can be seen to almost overlap with the greedy construction (\hyperref[alg:puncture]{Algorithm~\ref*{alg:puncture}}). The loss of \hyperref[alg:sym_puncture]{Algorithm~\ref*{alg:sym_puncture}} with respect to \hyperref[alg:puncture]{Algorithm~\ref*{alg:puncture}} is less than $0.025$ dB. The gain of \hyperref[alg:sym_puncture]{Algorithm~\ref*{alg:sym_puncture}} for polar H-ARQ is significant as seen in the next section.

\begin{figure}[!htb]
\centering
\includegraphics[scale=.42]{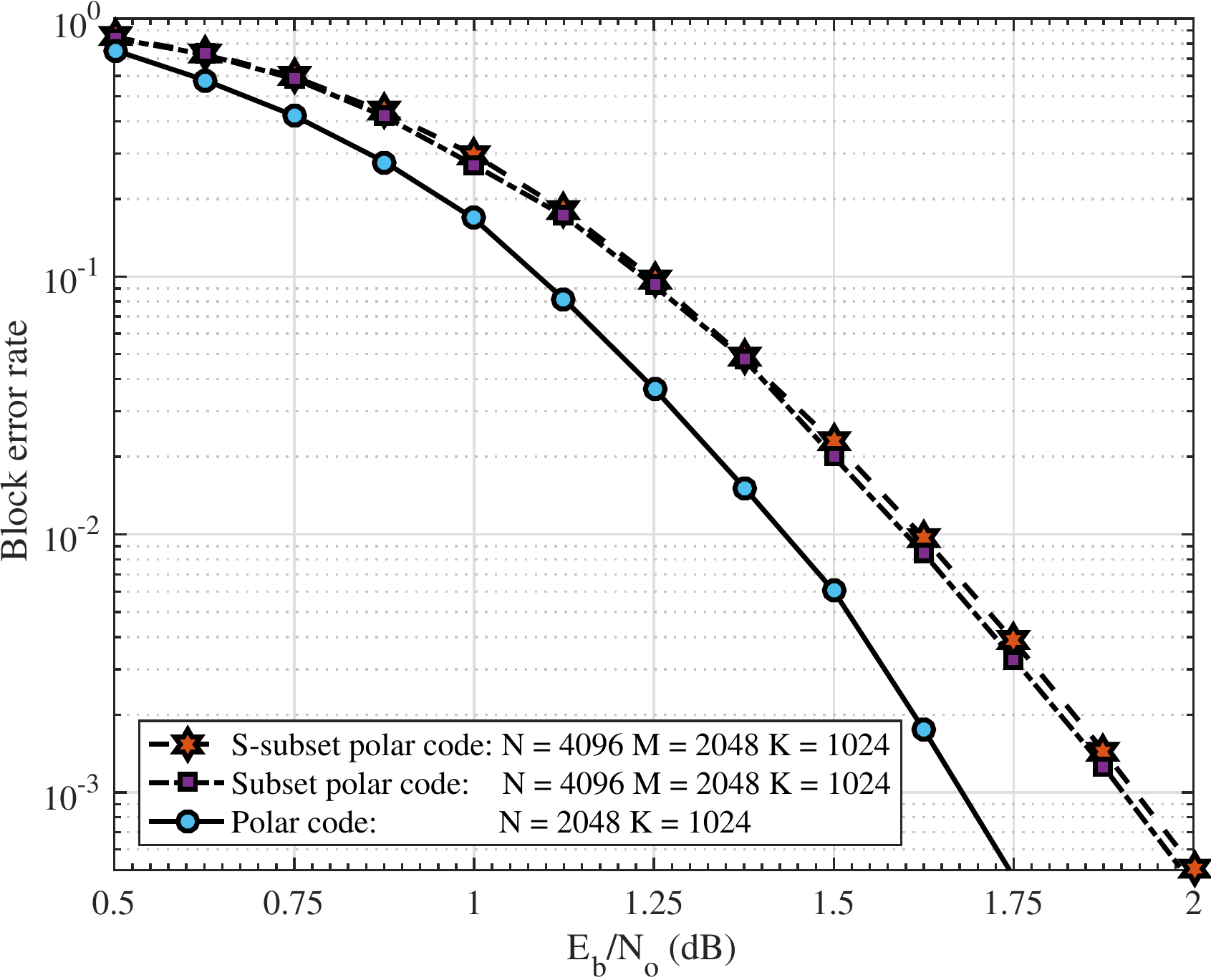}
\caption{Comparison of \hyperref[alg:puncture]{Algorithm~\ref*{alg:puncture}} and  \hyperref[alg:sym_puncture]{Algorithm~\ref*{alg:sym_puncture}}}\label{fig:symmetric}
\end{figure}

\subsection{Simulation results for polar H-ARQ}

Polar H-ARQ is supported by creating a subset polar code, $S$, as per \hyperref[alg:sym_puncture]{Algorithm~\ref*{alg:sym_puncture}}. We use value of $x = 4096, N = 4096, M = 2048$. We get two $(4096 \succeq 2048, 1024)$ subset polar codes, $S$ and $S~\underline{\oplus}~x$, which are used  as the two RVs. We note that more than two RVs can be generated by using this basic construction of equivalent subset codes $S$ and $S~\underline{\oplus}~x$. For example, two additional RVs can be generated as $S~\underline{\oplus}~y$ and $S ~\underline{\oplus}~x ~\underline{\oplus}~y$. The polar H-ARQ performance is demonstrated in \hyperref[fig:harq_performance]{Figure~\ref*{fig:harq_performance}}. It shows (i) two RVs, when decoded individually, have the same performance, and the performance is within $0.2$ dB (at $1 \%$ BLER) of a $(2048, 1024)$ polar code, and (ii) two RVs when decoded jointly, achieve performance of $(4096, 1024)$ mother polar code, which is about $0.75$ dB better than the $(2048, 1024)$ polar code.

\section{Comparison with prior work}\label{sec:prior_work}

The original polar construction in \cite{Arikan} used block length values that are powers of 2. Following \cite{Arikan}, there has been significant work to extend the construction to other block length values. There are two different types of extensions: (i) for example \cite{Korada}  extends polar constructions using an $l \times l$ kernel; (ii) \cite{Khamy}, \cite{Vera}, \cite{Wang} extend polar constructions by puncturing coded bits from an original polar code. Here, we limit the discussion to works related to the second approach, and more specifically to constructions that do not re-optimize the set of information bits. In addition, we discuss work related to H-ARQ with polar codes. 

\subsection{Subset polar codes with fixed information bits}

Proposed \hyperref[alg:puncture]{Algorithm~\ref*{alg:puncture}} is a small variation of the PPA algorithm proposed in \cite{Khamy} - the variation being update of the design $\epsilon$. The algorithm in \cite{Khamy} also uses Gaussian approximation of density evolution for code construction. Here, we use the simplified approach of using bounds on the Bhattacharya parameters which is numerically faster (see results in \cite{Vangala} that suggest that these two approaches have similar performance). In other work, \cite{Zhang} considered similar approach for code construction and proposed multiple algorithms. In particular, algorithm 4 in \cite{Zhang} is based on selecting coded bits that when punctured lead to zero capacity for the lowest capacity bit channels for the mother polar code which are frozen by code construction. A simulation comparison is shown in the \hyperref[fig:comparison]{Figure~\ref*{fig:comparison}} below. The results demonstrate that, for the parameters considered here, the proposed algorithm does about 0.25 and 0.75~dB better than algorithm with fixed $\epsilon$ and algorithm based on frozen bits respectively. One reason for worse performance of fixed $\epsilon$ algorithm is that the union bound is not tight for high code rate when evaluated for high value of $\epsilon$.\footnote{For the fixed $\epsilon$ scheme for subset polar code construction, we also tried using a low value of $\epsilon$ (or equivalently high value of SNR) for code design. This helps to keep the union bound tight, but the starting code itself has a significantly worse performance for the regime of interest ($\sim 1 \%$ BLER).}
\begin{figure}[!htb]
\centering
\includegraphics[scale=.42]{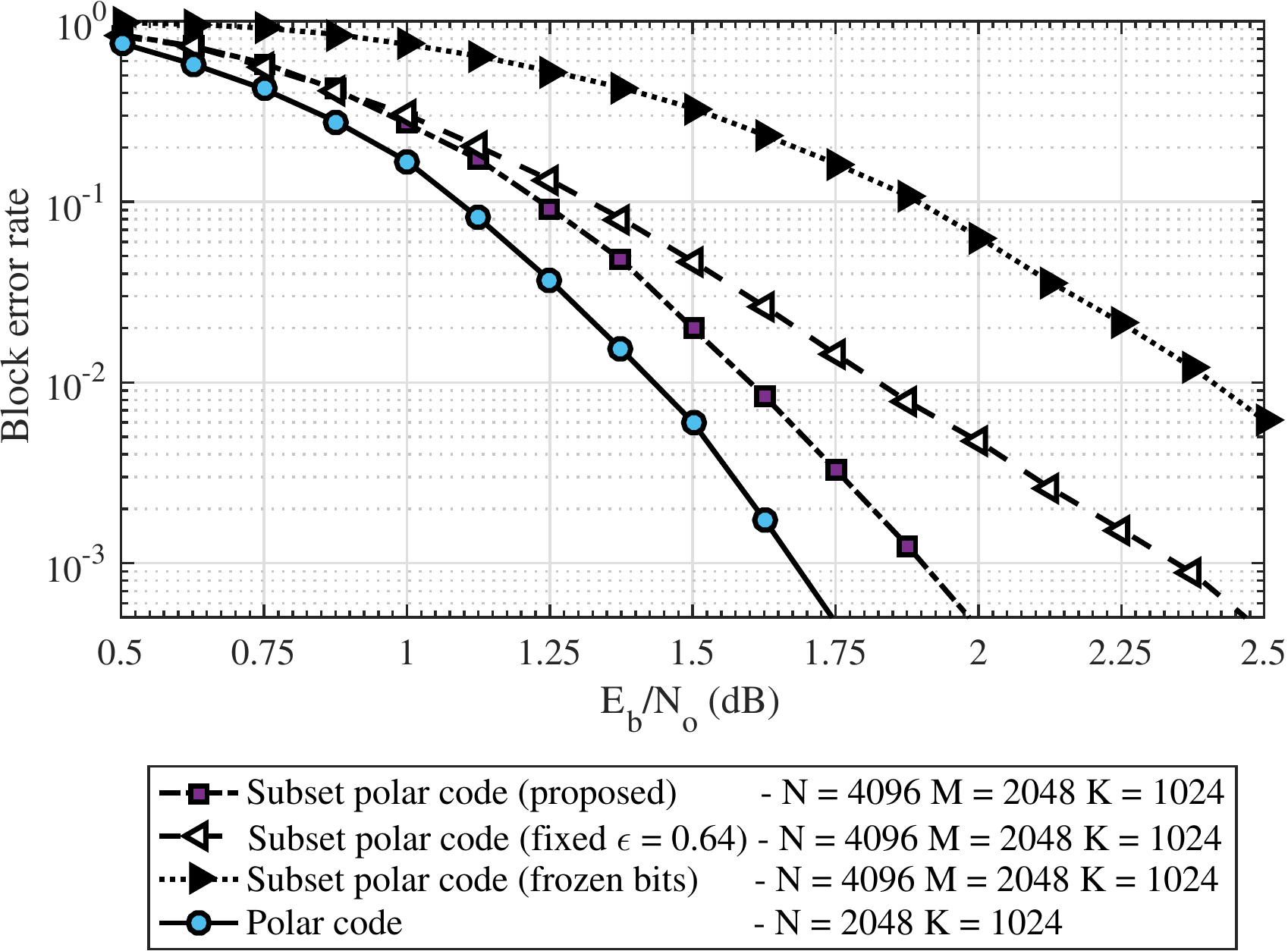}
\caption{Comparison of \hyperref[alg:puncture]{Algorithm~\ref*{alg:puncture}} with other approaches.}\label{fig:comparison}
\end{figure}

\subsection{Polar H-ARQ}

The work in \cite{Khamy} proposes a H-ARQ extension of the PPA algorithm by using the coded bits not transmitted during the first transmission. This proposal will have a similar performance when jointly decoding RVs, but it is unclear if each RV individually will have a good performance. This, for example, is important for a wireless system where the first transmission can experience a deep fade.

In addition to the traditional H-ARQ approach discussed in this paper (and in \cite{Khamy}),  another approach is proposed for polar H-ARQ in \cite{Li}, \cite{Maric}, \cite{Niu}, \cite{Chen}. At a high level, the motivation for these works is to be able to optimize code for each re-transmission by re-selecting the set of information bits. This improves performance of each transmission. However, since each transmission is effectively a different code, it is harder for the receiver to combine multiple transmissions. For example, the puncturing pattern in \cite{Niu} will lead to the $(4096 \succeq 2048, 1024)$ code to have the same performance as $(2048, 1024)$ code by re-optimizing the set of information bits. However, when jointly decoding multiple transmissions, the performance will be worse than performance of the base $(4096, 1024)$ polar code (e.g. see results in \cite{Li}).

\section{Equivalent polar subsets}\label{sec:equivalent}

Here, we prove \hyperref[thm1]{Theorem~\ref*{thm1}}. We start with a generalized notion of polarization, and then prove two lemmas regarding the generalization before proving \hyperref[thm1]{Theorem~\ref*{thm1}}. 
\subsection{Polarization with different distributions}
We generalize the notion of polarization in \cite{Arikan} to define polarization of two channels that are independent but not necessarily identically distributed (same as the definition of ``compound polar channels" in \cite{Hessam}). Let $W_1$ and $W_2$ be two binary input channels with output alphabets $Y_1$ and $Y_2$ respectively. We say a pair of binary input channels $W': X \rightarrow Y_1 \times Y_2$ and $W'': X \rightarrow Y_1 \times Y_2 \times X$ are obtained by single-step polarization transformation $(W_1, W_2)$ if:
\begin{eqnarray*}
W'(y_1, y_2 | u_1) & = & \sum_{u_2'} \frac{1}{2} W_1(y_1|u_1 \oplus u_2') W_2(y_2| u_2') \\
W''(y_1, y_2, u_1 | u_2) & = & \frac{1}{2} W_1(y_1|u_1 \oplus u_2). W_2(y_2| u_2) 
\end{eqnarray*}
We denote the polarization as $(W_1, W_2) \rightarrow (W', W'')$. 

One special case is when one of the channels is ``punctured". We define puncturing by use of an erasure channel, $\textit{E}$, that leads to an erasure with probability $1$. For example, for N = 2, if $S = \{ 1 \}$, the polarization transform for a subset polar code would involve $(W, \textit{E}) \rightarrow (W_{2,S}^{(1)}, W_{2,S}^{(2)})$. That is, the second realization of the channel $W$ is replaced by channel $\textit{E}$. We use this construction recursively to define channels $W_{N,S}^{(i)}$ starting with $N$ independent realizations of channel $W$ but replacing the realizations for indexes not in $S$ with $\textit{E}$.

\subsection{Lemma 1: order independence of polarization}

Given the general definition of polarization, we show that the order of channels does not matter for polarization for symmetric B-DMCs using the following lemma.
\begin{lemma}\label{lem:eq}
Let $W_1$ and $W_2$ be two independent symmetric B-DMCs, and let $(W_1, W_2) \rightarrow (W_1', W_1'')$ and $(W_2, W_1) \rightarrow (W_2', W_2'')$, then $W_1' ~ \sim ~  W_2'$ and $W_1'' ~ \sim ~ W_2'$.
\end{lemma}
\begin{proof}
To prove $W_1' \sim W_2'$, we use function $f : Y_1 \times Y_2 \rightarrow Y_2 \times Y_1$ to be $f(y_1, y_2) = (y_2, y_1)$. The equivalence can be verified by writing out the transition probabilities directly.

To prove $W_1'' \sim W_2''$, we use function $g : Y_1 \times Y_2 \times X \rightarrow Y_2 \times Y_1 \times X$ to be $g(y_1, y_2, u_1) = (u_1 \cdot y_2, u_1 \cdot y_1, u_1)$. Here, $u_1 \cdot y_2$ denotes $\pi_{u_1} (y_2)$ where $\pi_0$ is identity permutation and $\pi_1$ is the permutation such that (i) $\pi_1^{-1} = \pi_1$ and (ii) $W_2(y|1) = W_2(\pi_1(y)|0)$ which exists as the per the definition of a symmetric B-DMC. We use the same notation to denote the different permutations for alphabets $Y_1$ and $Y_2$. The function $g$ shows equivalence of $W_1''$ and $W_2''$ as follows:
\begin{eqnarray*}
W_1''(y_1, y_2, u_1 | u_2) & = &  \frac{1}{2} W_1(y_1|u_1 \oplus u_2) W_2(y_2| u_2) \\
 & = & \frac{1}{2} W_1(u_1 \cdot y_1| u_2) W_2(y_2| u_2) \\
 & = & \frac{1}{2} W_1(u_1 \cdot y_1| u_2) W_2(u_1 \cdot y_2| u_1 \oplus u_2) \\
 & = & W_2''(u_1 \cdot y_2, u_1 \cdot y_1, u_1| u_2)
\end{eqnarray*}
We note that some steps in the equation above use the result from \cite{Arikan} that $W(y|a\oplus x) = W(a \cdot y | x)$. Also, it can be checked that $f$ and $g$ are invertible functions. 
\end{proof}

\subsection{Lemma 2: polarization retains equivalence}
Next, we show that polarization of equivalent channels leads to equivalent polarized channels. 
\begin{lemma}\label{lem2}
Let $W_1$ and $W_2$ be two independent symmetric B-DMCs, and let $(W_1, W_2) \rightarrow (W', W'')$. Let $V_1$ and $V_2$ be two independent symmetric B-DMCs, and let $(V_1, V_2) \rightarrow (V', V'')$. If $V_1 \sim W_1$, and $V_2 \sim W_2$, then $V' \sim W'$ and $V'' \sim W''$.
\end{lemma}
\begin{proof}
Let $f_{W_1V_1}$ and $f_{W_2V_2}$ be the functions so that $V_1(f_{W_1V_1}(y_1)|x) = W_1(y_1|x)$ and $V_2(f_{W_2V_2}(y_2)|x) = W_2(y_2|x)$. Selecting the following functions: 
$f_{W''V''}: (y_1, y_2, u_1)  =  (f_{W_1V_1}(y_1), f_{W_2V_2}(y_2), u_1)$ and 
$f_{W'V'}: (y_1, y_2) = (f_{W_1V_1}(y_1), f_{W_2V_2}(y_2))$ shows equivalence of the polarized channels.
\end{proof}

\subsection{Proof of Theorem 1}

We use induction on $\log_2{N}$ - variable is denoted by $m$. 

For $m = 1$ ($N = 2$), the only two non-trivial and equivalent subsets are $S = \{ 1 \}$ and $T = \{ 2 \}$. The polarization transform would involve $(W, \textit{E}) \rightarrow (W_{2,S}^{(1)}, W_{2,S}^{(2)})$ and $(\textit{E}, W) \rightarrow (W_{2,T}^{(1)}, W_{2,T}^{(2)})$.  The equivalence of these channels follows from \hyperref[lem:eq]{Lemma~\ref*{lem:eq}} by using $W_1 ~ \leftarrow ~ W$ and $W_2 ~ \leftarrow ~ \textit{E}$.

Now, we assume \hyperref[thm1]{Theorem~\ref*{thm1}} is true for $m = n$ (or $N$).
 
We prove \hyperref[thm1]{Theorem~\ref*{thm1}} for $m = n + 1$ (or $2*N$). Let $S$ and  $T$ be the two equivalent subsets for $m = n+1$. We know that $S = T~\underline{\oplus}~x^{n+1}$. Let $x_{n+1}x_n...x_1$ be the binary representation of $x^{n+1}-1$, and let $x_nx_{n-1}...x_1$ be the binary representation on $x^{n} - 1$. To reduce the problem to $m = n$, we define:
\begin{eqnarray*}
S_1 & = & \{ s | s \in S, s \leq N \},~~ S_2 ~ = ~ S \setminus S_1   \\
T_1 & = & \{ t | t \in T, t \leq N \},~~ T_2 ~ = ~ T \setminus T_1 
\end{eqnarray*} 
Depending on the value of MSB of $x^{n+1} - 1$, $x_{n+1}$, and using the induction assumption, we have:
\begin{eqnarray*}
x_{n+1} = 0 ~ \Rightarrow ~ S_1 ~ = ~ T_1~\underline{\oplus}~x^{n}, & S_2 ~ = ~ T_2~\underline{\oplus}~x^{n}, \\
 \Rightarrow  W_{N,S_1}^{(i)}  \sim  W_{N,T_1}^{(i)}, & W_{N,S_2}^{(i)} \sim  W_{N,T_2}^{(i)}  \\
x_{n+1} = 1 ~\Rightarrow ~ S_1 ~ = ~ T_2  ~\underline{\oplus}~x^{n}, &  S_2 ~ = ~ T_1 ~\underline{\oplus}~x^{n}, \\
\Rightarrow  W_{N,S_1}^{(i)}  \sim  W_{N,T_2}^{(i)}, & W_{N,S_2}^{(i)} \sim  W_{N,T_1}^{(i)} 
\end{eqnarray*}
We next use \hyperref[lem:eq]{Lemma~\ref*{lem:eq}} and \hyperref[lem2]{\ref*{lem2}} with the following parameters:
\begin{eqnarray*}
u_1 ~ \leftarrow ~ u_{2i-1} && u_2 \leftarrow ~ u_{2i}; \\
W_1 ~ \leftarrow ~ W_{N,S_1}^{(i)}  && W_2 ~ \leftarrow ~ W_{N,S_2}^{(i)}; \\
W_1' ~ \leftarrow ~ W_{2N,S}^{(2*i-1)}  && W_1'' ~ \leftarrow ~ W_{2N,S}^{(2*i)}; \\
y_1 ~ \leftarrow ~ (Y_{S_1}, u_{1,o}^{2i-2} \oplus u_{1,e}^{2i-2}) && y_2 ~ \leftarrow ~ (Y_{S_2}, u_{1,e}^{2i-2}); \\
V_1 ~ \leftarrow ~ W_{N,T_1}^{(i)}  && V_2 ~ \leftarrow ~ W_{N,T_2}^{(i)}; \\
V_1' ~ \leftarrow ~ W_{2N,T}^{(2*i-1)}  && V_1'' ~ \leftarrow ~ W_{2N,T}^{(2*i)}; \\
y_1 ~ \leftarrow ~ (Y_{T_1}, u_{1,o}^{2i-2} \oplus u_{1,e}^{2i-2}) && y_2 ~ \leftarrow ~ (Y_{T_2}, u_{1,e}^{2i-2}).
\end{eqnarray*}
To complete the proof, we note that $(W_{N,S_1}^{(i)}, W_{N,S_2}^{(i)}) \rightarrow (W_{2N,S}^{(2*i-1)}, W_{2N,S}^{(2*i)})$, and $(W_{N,T_1}^{(i)}, W_{N,T_2}^{(i)}) \rightarrow (W_{2N,T}^{(2*i-1)}, W_{2N,T}^{(2*i)})$. Further, as per \hyperref[lem:eq]{Lemma~\ref*{lem:eq}}, the order of parameters does not matter for polarization, and as per the induction step either (i) $W_{N,S_1}^{(i)}  \sim  W_{N,T_1}^{(i)}, ~ W_{N,S_2}^{(i)} \sim  W_{N,T_2}^{(i)}$ or (ii) $W_{N,S_1}^{(i)}  \sim  W_{N,T_2}^{(i)},~ W_{N,S_2}^{(i)} \sim  W_{N,T_1}^{(i)}$.Therefore, using \hyperref[lem2]{Lemma~\ref*{lem2}} we conclude that  $W_{2N,S}^{(2*i-1)} \sim  W_{2N,T}^{(2*i-1)}$ and $W_{2N,S}^{(2*i)} \sim  W_{2N,T}^{(2*i)}$.

\end{document}